\documentclass[11pt,letterpaper]{article}
\usepackage{fullpage}

\usepackage{graphicx,subfigure}

\def\showauthornotes{0}

\def\showkeys{0}

\def\showcolorlinks{1}
\def\usemicrotype{1}

\def\arxivmode{0}
\def\fastmode{0}

\newcommand{\llangle}{\left\langle}
\newcommand{\rrangle}{\right\rangle}

\ifnum\fastmode=0
\usepackage{etex}
\fi

\ifnum\fastmode=0
\usepackage[l2tabu, orthodox]{nag}
\fi

\usepackage{xspace,enumerate}

\usepackage[dvipsnames]{xcolor}

\ifnum\fastmode=0
\usepackage[T1]{fontenc}
\usepackage[full]{textcomp}
\fi

\usepackage[american]{babel}

\usepackage{mathtools}

\usepackage{amsthm}

\newtheorem{theorem}{Theorem}[section]
\newtheorem*{theorem*}{Theorem}

\newtheorem*{proposition*}{Proposition}
\newtheorem{lemma}[theorem]{Lemma}
\newtheorem*{lemma*}{Lemma}
\newtheorem{corollary}[theorem]{Corollary}
\newtheorem*{conjecture*}{Conjecture}

\newtheorem*{fact*}{Fact}

\newtheorem*{exercise*}{Exercise}

\newtheorem*{hypothesis*}{Hypothesis}

\theoremstyle{definition}

\newtheorem{exercise-easy}[theorem]{Exercise}
\newtheorem{exercise-med}[theorem]{Exercise}
\newtheorem{exercise-hard}[theorem]{Exercise$^\star$}

\newtheorem*{claim*}{Claim}

\newtheorem{remark}[theorem]{Remark}
\newtheorem*{remark*}{Remark}

\newtheorem*{observation*}{Observation}

\ifnum\arxivmode=0
\usepackage{newpxtext} %
\usepackage{textcomp} %
\usepackage[varg,bigdelims]{newpxmath}
\usepackage[cal=cm,scr=rsfso]{mathalfa}%
\usepackage{bm} %
\let\mathbb\vvmathbb
\fi

\ifnum\arxivmode=1
\usepackage[varg]{pxfonts} %
\usepackage{textcomp} %
\usepackage[scr=rsfso]{mathalfa}%
\usepackage{bm} %
\fi

\ifnum\showkeys=1
\usepackage[color]{showkeys}
\fi

\makeatletter
\@ifpackageloaded{hyperref}
{}
{\usepackage[pagebackref]{hyperref}}

\definecolor{bleudefrance}{rgb}{0.01, 0.1, 1.0}
\definecolor{azure}{rgb}{0.0, 0.5, 1.0}

\ifnum\showcolorlinks=1
\hypersetup{
pagebackref=true,
colorlinks=true,
urlcolor=blue,
linkcolor=blue,
citecolor=OliveGreen}
\fi

\ifnum\showcolorlinks=0
\hypersetup{
pagebackref=true,
colorlinks=false,
pdfborder={0 0 0}}
\fi

\usepackage{prettyref}

\newcommand{\savehyperref}[2]{\texorpdfstring{\hyperref[#1]{#2}}{#2}}

\newrefformat{eq}{\savehyperref{#1}{\textup{(\ref*{#1})}}}
\newrefformat{lem}{\savehyperref{#1}{Lemma~\ref*{#1}}}
\newrefformat{def}{\savehyperref{#1}{Definition~\ref*{#1}}}
\newrefformat{thm}{\savehyperref{#1}{Theorem~\ref*{#1}}}
\newrefformat{cor}{\savehyperref{#1}{Corollary~\ref*{#1}}}
\newrefformat{cha}{\savehyperref{#1}{Chapter~\ref*{#1}}}
\newrefformat{sec}{\savehyperref{#1}{Section~\ref*{#1}}}
\newrefformat{app}{\savehyperref{#1}{Appendix~\ref*{#1}}}
\newrefformat{ass}{\savehyperref{#1}{Assumption~\ref*{#1}}}
\newrefformat{tab}{\savehyperref{#1}{Table~\ref*{#1}}}
\newrefformat{fig}{\savehyperref{#1}{Figure~\ref*{#1}}}
\newrefformat{hyp}{\savehyperref{#1}{Hypothesis~\ref*{#1}}}
\newrefformat{alg}{\savehyperref{#1}{Algorithm~\ref*{#1}}}
\newrefformat{rem}{\savehyperref{#1}{Remark~\ref*{#1}}}
\newrefformat{item}{\savehyperref{#1}{Item~\ref*{#1}}}
\newrefformat{step}{\savehyperref{#1}{step~\ref*{#1}}}
\newrefformat{conj}{\savehyperref{#1}{Conjecture~\ref*{#1}}}
\newrefformat{fact}{\savehyperref{#1}{Fact~\ref*{#1}}}
\newrefformat{prop}{\savehyperref{#1}{Proposition~\ref*{#1}}}
\newrefformat{prob}{\savehyperref{#1}{Problem~\ref*{#1}}}
\newrefformat{claim}{\savehyperref{#1}{Claim~\ref*{#1}}}
\newrefformat{relax}{\savehyperref{#1}{Relaxation~\ref*{#1}}}
\newrefformat{red}{\savehyperref{#1}{Reduction~\ref*{#1}}}
\newrefformat{part}{\savehyperref{#1}{Part~\ref*{#1}}}
\newrefformat{ex}{\savehyperref{#1}{Exercise~\ref*{#1}}}
\newrefformat{property}{\savehyperref{#1}{Property~\ref*{#1}}}

\newcommand{\Sref}[1]{\hyperref[#1]{\S\ref*{#1}}}

\usepackage{nicefrac}

\ifnum\fastmode=0
\ifnum\usemicrotype=1
\usepackage{microtype}
\fi
\fi

\ifnum\showauthornotes=2
\newcommand{\mynotes}[1]{{\sffamily\small\color{teal}{#1}}\medskip}
\newcommand{\Authornote}[2]{{\sffamily\small\color{blue}{[#1: #2]}}\medskip}
\newcommand{\Authornotecolored}[3]{{\sffamily\small\color{#1}{[#2: #3]}}}
\newcommand{\Authorcomment}[2]{{\sffamily\small\color{gray}{[#1: #2]}}}
\newcommand{\Authorstartcomment}[1]{\sffamily\small\color{gray}[#1: }

\newcommand{\Authorfnote}[2]{\footnote{\color{red}{#1: #2}}}
\newcommand{\Authormarginmark}[1]{\marginpar{\textcolor{red}{\fbox{\Large #1:!}}}}
\newcommand{\myexplain}[1]{{\sffamily\small\color{red}{\noindent [Explanation:\medskip\newline \begin{quote}#1\hfill]\end{quote}}}\medskip}
\newcommand{\explain}[1]{{\sffamily\small\color{red}{#1}}\medskip}

\else

\newcommand{\mynotes}[1]{}
\newcommand{\Authornote}[2]{}
\newcommand{\Authornotecolored}[3]{}
\newcommand{\Authorcomment}[2]{}
\newcommand{\Authorstartcomment}[1]{}

\newcommand{\Authorfnote}[2]{}
\newcommand{\Authormarginmark}[1]{}
\newcommand{\myexplain}[1]{}
\newcommand{\explain}[1]{}

\ifnum\showauthornotes=1
\renewcommand{\myexplain}[1]{{\sffamily\small\color{red}{\noindent \begin{quote}{\bf Explanation:} \medskip\newline #1\end{quote}}}\medskip}
\fi

\fi

\usepackage{boxedminipage}

\newcommand{\Esymb}{\mathbb{E}}
\newcommand{\Psymb}{\vvmathbb{P}}

\DeclareMathOperator*{\E}{\Esymb}

\DeclareMathOperator*{\ProbOp}{\Psymb}

\renewcommand{\Pr}{\ProbOp}

\newcommand{\textparen}[1]{\text{(#1)}}

\ifx\because\undefined
\newcommand{\because}[1]{\textparen{because #1}}
\else
\renewcommand{\because}[1]{\textparen{because #1}}
\fi

\newcommand{\seteq}{\mathrel{\mathop:}=}

\newcommand{\bigmid}{~\big|~}
\newcommand{\Bigmid}{~\Big|~}

\newcommand\bdot\bullet

\ifx\mathds\undefined %
\newcommand{\Ind}{\mathbb I}
\else
\newcommand{\Ind}{\mathds 1}
\fi

\newcommand{\R}{\vvmathbb R}

\newcommand{\cA}{\mathcal A}

\newcommand{\cC}{\mathcal C}
\newcommand{\cD}{\mathcal D}
\newcommand{\cE}{\mathcal E}

\newcommand{\cL}{\mathcal L}

\newcommand{\sfD}{\mathsf D}

\newcommand{\sfK}{\mathsf K}

\newcommand{\sfN}{\mathsf N}

\renewcommand{\leq}{\leqslant}
\renewcommand{\le}{\leqslant}
\renewcommand{\geq}{\geqslant}
\renewcommand{\ge}{\geqslant}

\let\epsilon=\varepsilon

\numberwithin{equation}{section}

\newcommand\MYcurrentlabel{xxx}

\newcommand{\MYstore}[2]{%
  \global\expandafter \def \csname MYMEMORY #1 \endcsname{#2}%
}

\newcommand{\MYload}[1]{%
  \csname MYMEMORY #1 \endcsname%
}

\newcommand{\MYnewlabel}[1]{%
  \renewcommand\MYcurrentlabel{#1}%
  \MYoldlabel{#1}%
}

\newcommand{\MYdummylabel}[1]{}

\newcommand{\torestate}[1]{%
  \let\MYoldlabel\label%
  \let\label\MYnewlabel%
  #1%
  \MYstore{\MYcurrentlabel}{#1}%
  \let\label\MYoldlabel%
}

\newcommand{\restatetheorem}[1]{%
  \let\MYoldlabel\label
  \let\label\MYdummylabel
  \begin{theorem*}[Restatement of \prettyref{#1}]
    \MYload{#1}
  \end{theorem*}
  \let\label\MYoldlabel
}

\newcommand{\restatelemma}[1]{%
  \let\MYoldlabel\label
  \let\label\MYdummylabel
  \begin{lemma*}[Restatement of \prettyref{#1}]
    \MYload{#1}
  \end{lemma*}
  \let\label\MYoldlabel
}

\newcommand{\restateprop}[1]{%
  \let\MYoldlabel\label
  \let\label\MYdummylabel
  \begin{proposition*}[Restatement of \prettyref{#1}]
    \MYload{#1}
  \end{proposition*}
  \let\label\MYoldlabel
}

\newcommand{\restatefact}[1]{%
  \let\MYoldlabel\label
  \let\label\MYdummylabel
  \begin{fact*}[Restatement of \prettyref{#1}]
    \MYload{#1}
  \end{fact*}
  \let\label\MYoldlabel
}

\newcommand{\restate}[1]{%
  \let\MYoldlabel\label
  \let\label\MYdummylabel
  \MYload{#1}
  \let\label\MYoldlabel
}

\newcommand{\addreferencesection}{
  \phantomsection
\ifnum\stocmode=0
  \addcontentsline{toc}{section}{References}
\else
  \addcontentsline{toc}{section}{References \hspace*{1in} --------- End of extended abstract ---------}
\fi

}

\newcommand{\e}{\epsilon}
\newcommand{\eps}{\epsilon}

\let\origparagraph\paragraph
\renewcommand{\paragraph}[1]{\vspace*{-2pt}\origparagraph{#1.}}

\allowdisplaybreaks

\usepackage{paralist}

\usepackage{comment}

\let\pref=\prettyref

\newcommand{\diam}{\mathrm{diam}}
\newcommand{\dmid}{\,\|\,}

\newcommand{\vertiii}[1]{{\left\vert\kern-0.25ex\left\vert\kern-0.25ex\left\vert #1 
          \right\vert\kern-0.25ex\right\vert\kern-0.25ex\right\vert}}

\renewcommand{\Ind}{\vvmathbb{1}}
\let\1\Ind

\DeclareMathOperator{\argmin}{\mathrm{argmin}}

\usepackage{accents}
\usepackage{enumitem}
\usepackage[normalem]{ulem}

\newcommand{\K}{\mathsf{K}}

\newcommand{\rt}{\vvmathbb{r}}

\newcommand{\depth}{\mathfrak{D}}
\newcommand{\tnorm}[1]{\left\|#1\right\|_{\ell_1(w)}}

\newcommand{\nabA}[1]{\nabla}

\newcommand{\scost}{\mathsf{S}}
\newcommand{\mcost}{\mathsf{M}}
\newcommand{\cost}{\mathrm{cost}}

\newcommand{\sfp}{\mathsf{p}}

\newcommand{\vvr}{\vvmathbb{r}}

\newcommand{\tD}{\tilde{\sfD}}

\begin{document}

\title{Pure entropic regularization for metrical task systems}
\author{Christian Coester \\ {\small University of Oxford} \and James R. Lee \\ {\small University of Washington}}
\date{}

\maketitle

\begin{abstract}
We show that on every $n$-point HST metric, there is a randomized online
algorithm for metrical task systems (MTS) that is $1$-competitive for service costs and $O(\log n)$-competitive for movement costs.
In general, these refined guarantees are optimal up to the implicit constant.
While an $O(\log n)$-competitive algorithm for MTS on HST metrics was developed in \cite{BCLL19},
that approach could only establish an $O((\log n)^2)$-competitive ratio
when the service costs are required to be $O(1)$-competitive.
Our algorithm can be viewed as an instantiation of
online mirror descent with the regularizer derived from a multiscale conditional entropy.

In fact, our algorithm satisfies a set of even more refined guarantees; we are able to exploit this
property to combine it with known random embedding theorems and obtain, for {\em any} $n$-point metric space,
a randomized algorithm that is $1$-competitive for service costs and $O((\log n)^2)$-competitive
for movement costs.
\end{abstract}

\section{Introduction}

Let $(X,d_X)$ be a finite metric space with $|X|=n > 1$. The Metrical Task
Systems (MTS) problem, introduced in \cite{BLS92} is described as follows.
The input is a sequence $\langle c_t : X \rightarrow \R_+ : t \geq 1\rangle$ of
nonnegative cost functions on the state space $X$.
At every time $t$, an online algorithm maintains
a state $\rho_t \in X$.

The corresponding cost is the sum of a {\em service cost} $c_t(\rho_t)$ and a
{\em movement cost} $d_X(\rho_{t-1}, \rho_t)$.
Formally, an {\em online algorithm} is a sequence of mappings
$\bm{\rho} = \langle \rho_1, \rho_2, \ldots, \rangle$
where, for every $t \geq 1$,
$\rho_t : (\R_+^X)^t \to X$ maps a sequence of cost functions $\langle c_1, \ldots, c_t\rangle$
to a state. The initial state $\rho_0 \in X$ is fixed. The {\em total cost of the algorithm $\bm{\rho}$ in servicing $\bm{c} = \langle c_t : t \geq 1\rangle$} is defined as:
\[
\cost_{\bm{\rho}}(\bm{c}) \seteq \sum_{t \geq 1} \left[c_t\!\left(\rho_t(c_1,\ldots, c_t)\right) + d_X\!\left(\rho_{t-1}(c_1,\ldots, c_{t-1}), \rho_t(c_1,\ldots, c_t)\right)\right].
\]
The cost of the {\em offline optimum}, denoted $\cost^*(\bm{c})$, is the infimum of
$\sum_{t \geq 1} [c_t(\rho_t)+d_X(\rho_{t-1},\rho_t)]$ over {\em any} sequence
$\llangle \rho_t : t \geq 1\rrangle$ of states.
A {\em randomized online algorithm} $\bm{\rho}$  is said to be {\em $\alpha$-competitive}
if for every $\rho_0 \in X$, there is a constant $\beta > 0$ such that for all
cost sequences $\bm{c}$:
\[
\E\left[\cost_{\bm{\rho}}(\bm{c})\right] \leq \alpha \cdot \cost^*(\bm{c}) + \beta\,.
\]

For the $n$-point uniform metric,
a simple coupon-collector argument shows that the competitive ratio is $\Omega(\log n)$, and this is tight \cite{BLS92}. A long-standing conjecture is that this $\Theta(\log n)$ competitive ratio holds for an arbitrary $n$-point metric space.
The lower bound has almost been established \cite{BBM06,BLMN05};
for any $n$-point metric space, the competitive ratio is $\Omega(\log n / \log \log n)$.
Following a long sequence of works (see, e.g., \cite{Sei99,BKRS00,BBBT97,Bar96,FM03,FRT04}),
an upper bound of $O((\log n)^2)$ was shown in \cite{BCLL19}.

\paragraph{Relation to adversarial multi-arm bandits}
MTS is naturally related to the adversarial setting of the
classical multi-arm bandits model in sequential decision making,
and provides a very general framework for ``bandits with switching costs.''
Unlike in the setting of regret minimization, where one competes
against the best static strategy in hindsight (see, e.g., \cite{BN12}),
competitive analysis compares the performance of an online algorithm
to the best {\em dynamical} offline algorithm.

Thus this model emphasizes the importance of
an adaptivity in the face of changing environments.
For MTS, the online algorithm has {\em full information}: access to
the complete cost function $c_t$ is available when deciding on a point $\rho_t(c_1,\ldots,c_t) \in X$
at which to play.
And yet
one of the fascinating relationships 
between MTS and adversarial bandits is the parallel between
adaptivity---being willing to ``try out'' new strategies---and
the classical exploration/exploitation tradeoff that
occurs in models where one only has access to partial information about
the loss functions.

\paragraph{HST metrics}
The methods of \cite{BBN12} show that the competitive ratio for MTS is $O(\log n)$
on weighted star metrics.  Recently, the authors of \cite{BCLL19} generalized this
result by designing an algorithm with competitive ratio $O(\depth_T \log n)$
on any weighted $n$-point tree metric with combinatorial depth $\depth_T$.
We now discuss a special class of metrics.

Let $T=(V,E)$ be a finite tree with root $\rt$ and vertex weights
$\{ w_u > 0 : u \in V\}$,
let $\cL \subseteq V$ denote the leaves of $T$,
and suppose that the vertex weights on $T$ are non-increasing along root-leaf paths.
Consider the metric space $(\cL,d_T)$, where $d_T(\ell,\ell')$ is the weighted length of the path connecting $\ell$ and $\ell'$ when the edge from a node $u$ to its parent is $w_u$.
We will use $\depth_T$ for the combinatorial (i.e., unweighted) depth of $T$.

$(\cL,d_T)$ is called an {\em HST metric} (or, equivalently
for finite metric spaces, an {\em ultrametric}).
If, for some $\tau > 1$, the weights on $T$ satisfy the stronger inequality $w_v \leq w_u/\tau$ whenever
$v$ is a child of $u$, the space $(\cL,d_T)$ is said to be a {\em $\tau$-HST metric.}
Such metric spaces play a special role in MTS since every $n$-point
metric space can be probabilistically approximated by a distribution over
such spaces \cite{Bar96,FRT04}.  Indeed, the $O((\log n)^2)$-competitive ratio
for general metric spaces established in \cite{BCLL19} is a consequence
of their $O(\log n)$-competitive algorithm for HSTs.

\subsection{Refined guarantees}

The authors of \cite{BBN10} observe that there is a more refined way to analyze
competive algorithms for MTS.
For a randomized online algorithm $\bm{\rho}$ and a cost sequence $\bm{c}$, we denote, respectively, $\scost_{\bm{\rho}}(\bm{c})$
and $\mcost_{\bm{\rho}}(\bm{c})$ for the (expected) service cost and movement cost, that is
\[
\scost_{\bm{\rho}}(\bm{c}) \seteq \E \sum_{t \geq 1} c_t(\rho_t) \quad \text{and}\quad \mcost_{\bm{\rho}}(\bm{c}) \seteq \E \sum_{t\geq1}d_X(\rho_{t-1}, \rho_t) \,.
\]
If there are numbers $\alpha,\alpha',\beta,\beta' > 0$ such that
for every cost $\bm{c}$, it holds that
\begin{align*}
\scost_{\bm\rho}(\bm{c}) &\leq \alpha \cdot \cost^*(\bm{c})+\beta \\
\mcost_{\bm\rho}(\bm{c}) &\leq \alpha' \cdot \cost^*(\bm{c})+\beta',
\end{align*}
one says that $\bm{\rho}$ is {\em $\alpha$-competitive for service costs} and {\em $\alpha'$-competitive for movement costs.}

In \cite{BBN10}, it is shown that on every $n$-point HST metric, and
for every $\e > 0$, there is an online algorithm that is simultaneously
$(1+\e)$-competitive for service costs and $O((\log (n/\e))^2)$-competitive for movement costs.
The authors of \cite{BCLL19} improve this slightly to show that actually
there is an online algorithm that is simultaneously $1$-competitive for service costs
and $O((\log n)^2)$-competitive for movement costs.
We obtain the optimal refined guarantees.

\begin{theorem} \label{thm:refined}
	On any $n$-point HST metric $X$, there
	is a randomized  online algorithm that is $1$-competitive for service costs and $O(\log n)$-competitive for movement costs.
\end{theorem}

\begin{remark}[Optimality of the refined guarantees]
Any finitely competitive algorithm for MTS on an $n$-point uniform metric
cannot be better than $\Omega(\log n)$-competitive for movement costs,
regardless of its competitive ratio for service costs. This is because this lower bound holds even if the cost functions only take values $0$ and $\infty$. Moreover, it cannot be better than $1$-competitive for service costs, regardless of its competitive ratio for movement costs. To see this, consider the case where each cost function is the constant function $1$.
\end{remark}

\paragraph{Finely competitive guarantees}
Suppose that for some numbers $\alpha_0,\alpha_1,\gamma,\beta,\beta' > 0$,
a randomized online algorithm $\bm\rho$ satisfies, for every cost $\bm{c}$ and
{\em every} offline algorithm $\bm{\rho}^*$:
\begin{align}
   \scost_{\bm\rho}(\bm{c}) &\leq \alpha_0 \scost_{\bm\rho^*}(\bm{c}) + \alpha_1 \mcost_{\bm\rho^*}(\bm{c}) + \beta \label{eq:finesc}\\
   \mcost_{\bm\rho}(\bm{c}) &\leq \gamma \scost_{\bm{\rho}}(\bm{c}) + \beta'\,.\label{eq:finemc}
\end{align}
In this case, we say that {\em $\bm{\rho}$ is $(\alpha_0,\alpha_1,\gamma)$-finely competitive.}
We establish the following.

\begin{theorem}\label{thm:super}
   On any $n$-point HST metric $X$, for every $\kappa \geq 1$, there is an online randomized algorithm $\bm{\rho}$ that is
   $\left(1,1/\kappa,O(\kappa \log n)\right)$-finely competitive.
   In fact, one can take $\beta=0$ and $\beta' \leq O(\kappa \diam(X))$.
\end{theorem}

Combined with the random embedding from \cite{FRT04}, this yields the following consequence
for general $n$-point metric spaces.

\begin{corollary}
   On any $n$-point metric space, there is an online randomized algorithm that is
   $1$-competitive for service costs and $O((\log n)^2)$-competitive for movement costs.
\end{corollary}

\begin{proof}
Consider an $n$-point metric space $(X,d_X)$.  It is known \cite{FRT04} that there
exists a random HST metric $(T,d_T)$ so that $\cL(T)=X$ and for all $x,y \in X$:
\begin{enumerate}
   \item $\Pr[d_T(x,y) \geq d_X(x,y)]=1$,
   \item $\E[d_T(x,y)] \leq D\cdot d_X(x,y)$,
\end{enumerate}
and $D \leq O(\log n)$.

Let $\bm\rho_T$ be the randomized algorithm for $(T,d_T)$ guaranteed by \pref{thm:super} with $\kappa=D$.
Let $\bm\rho$ denote the algorithm that results from sampling $(T,d_T)$ and then using $\bm\rho_T$.
We use $\mcost^T$ to denote movement cost measured in $d_T$ and $\mcost^X$ for movement cost measured in $d_X$.

Then for any cost $\bm{c}$ and any offline algorithm $\bm\rho^*$, we have
\begin{align*}
   \scost_{\bm\rho}(\bm{c})  = \E[\scost_{\bm\rho_T}(\bm{c})] &\leq \scost_{\bm\rho^*}(\bm{c}) + \kappa^{-1} \E[\mcost^T_{\bm\rho^*}(\bm{c})] + O(1)\\
&\leq \scost_{\bm\rho^*}(\bm{c}) + \kappa^{-1} D \mcost^X_{\bm\rho^*}(\bm{c}) + O(1) \\
&= \scost_{\bm\rho^*}(\bm{c}) + \mcost^X_{\bm\rho^*}(\bm{c}) + O(1)\,,
\end{align*}
and
\begin{align*}
   \mcost_{\bm{\rho}}^X(\bm{c}) = \E[\mcost^X_{\bm\rho_T}(\bm{c})] \leq \E[\mcost^T_{\bm\rho_T}(\bm{c})] \leq O(\kappa \log n) \E[\scost_{\bm\rho_T}(\bm{c})]\,+ O(1),
\end{align*}
completing the proof.
\end{proof}

\subsection{The fractional model on trees}

We will work in the following deterministic fractional setting, which is equivalent to the randomized integral setting described earlier (see \cite[\S 2]{BCLL19}). The state of a fractional algorithm is given by a point in the polytope
\begin{align}\label{eq:mtsK}
   \K_T &\seteq \left\{ x \in \R_+^{V} : x_{\rt} = 1,\ x_u = \sum_{v \in \chi(u)} x_v \quad \forall u \in V \setminus \cL\right\},
\end{align}
where we use $\chi(u)$ for the set of children of $u$ in $T$. For $u \ne\rt$, we will also write $\sfp(u)$ for the parent of $u$ in $T$.

A state $x\in\K_T$ corresponds to the situation that the state of a randomized integral algorithm is a leaf descendant of $u$ with probability $x_u$. Note that $\K_T$ is simply an affine encoding of the probability simplex on $\cL$. In the fractional setting, changing from state $x$ to $x'$ incurs movement cost $\tnorm{x-x'}$, where
\[
\tnorm{z} \seteq \sum_{u \in V} w_u |z_u|
\]
denotes the weighted $\ell_1$-norm on $\R^{V}$.

\subsection{Mirror descent, metric filtrations, and regularization}

Following \cite{BCLL19}, our algorithm is based on the mirror descent framework as established in \cite{BCLLM18}.
This is a method for regularized online convex optimization,
an approach that was previously explored for competitive analysis in \cite{ABBS10,BCN14}.

A central component of mirror descent is choosing the appropriate mirror map
(which we will often refer to as the ``regularizer'').  This is a strictly convex function
$\Phi : \K_T \to \R$ that endows $\K_T$ with a geometric (Riemannian) structure,
specifying how to perform constrained vector flow.  In other words,
it specifies how one can move in a preferred direction while remaining inside $\K_T$.

The paper \cite{BCLL19} employs the following regularizer:
\begin{equation}\label{eq:reg0}
\Phi_0(x) \seteq \frac{1}{\eta} \sum_{u \in V \setminus \{\rt\}} w_u \left(x_u+\delta_u\right) \log \left(x_u+\delta_u\right)\,,
\end{equation}
with $\eta =\Theta(\log |\cL|)$ and $\delta_u = |\cL_u|/|\cL|$, where $\cL_u$ is the set of leaves in the subtree rooted at $u$.

\subsubsection{Metric filtrations}
It is straightforward that one can think of $\Phi_0$ as a type of multiscale entropy (this is the
{\em negative} of the associated Shannon entropy, since we use the analyst's convention
that the entropy is convex).
To understand this notion, let us forget momentarily the weights on $T$.  Then
the structure of $T$ gives a natural filtration over probability measures on the leaves $\cL$.
Suppose that $\bm{X}$ is a random variable taking values in $\cL$ and,
for $u \in V$, denote by $\cE_u$ the event $\{\bm{X} \in \cL_u\}$.
Then the chain rule for Shannon entropy yields
\[
   \sum_{\ell\in\cL} \Pr[\cE_\ell]\log \frac{1}{\Pr[\cE_{\ell}]} = 
   \sum_{u\in V\setminus\{\rt\}}
   \Pr[\cE_u] \log \frac{\Pr[\cE_{\sfp(u)}]}{\Pr[\cE_u]}.
\]

If we now imagine that uncertainty at higher scales is more costly
than uncertainty at lower scales, then
we might define an analogous {\em weighted} entropy by
\begin{equation}\label{eq:weighted-ent}
   \sum_{u\in V\setminus\{\rt\}}
   w_u \Pr[\cE_u] \log \frac{\Pr[\cE_{\sfp(u)}]}{\Pr[\cE_{u}]}.
\end{equation}
Such a notion is natural in the context of ``metric learning'' problems.

Ignoring the $\{\delta_u\}$ values for a moment, consider that \eqref{eq:reg0} is not analogous
to \eqref{eq:weighted-ent}.  Indeed, it corresponds to the quantity
\begin{equation}\label{eq:badent}
\sum_{u\in V\setminus\{\rt\}}w_u
\Pr[\cE_u] \log \frac{1}{\Pr[\cE_u]},
\end{equation}
and now one can see a fundamental reason why the algorithm associated to \eqref{eq:reg0}
only achieves an $O(\depth_T \log n)$ competitive ratio, where $\depth_T$ is the combinatorial
depth of $T$:  The quantity \eqref{eq:badent} {\em overmeasures} the metric uncertainty.

Suppose that $\bm{X}$ is a uniformly random leaf.  Then
$\sum_{\ell\in\cL} \Pr[\cE_\ell] \log \frac{1}{\Pr[\cE_{\ell}]} = \log n$, where $n=|\cL|$.
But, in general, one could have $\sum_{u \in V} \Pr[\cE_u] \log \frac{1}{\Pr[\cE_u]} \geq \Omega(\depth_T \log n)$.
This fact was not lost on the authors of \cite{BCLL19},
but they bypass the problem by combining mirror descent on stars with a recursive
composition method called ``unfair gluing.''

\subsubsection{Multiscale conditional entropy}

We employ a regularizer that is a more faithful analog of \eqref{eq:weighted-ent}:
\begin{equation}\label{eq:reg1}
\Phi(x) \seteq \sum_{u \in V \setminus \{\rt\}} \frac{w_u}{\eta_u} \left(x_u + \delta_u x_{\sfp(u)}\right) \log \left(\frac{x_u}{x_{\sfp(u)}}+\delta_u\right),
\end{equation}
where $\sfp(u)$ denotes the parent of $u$.

If one ignores the additional parameters $\{\eta_u \geq 1,\delta_u > 0\}$, this is precisely the negative
weighted Shannon entropy written according to the chain rule.
Here, we set
\begin{align}
\theta_u &\seteq \frac{|\cL_u|}{|\cL_{\sfp(u)}|} \label{eq:thetadef} \\
\eta_u &\seteq 1+\log(1/\theta_u) \label{eq:etadef} \\
\delta_u &\seteq \theta_u/\eta_u\,. \label{eq:deltadef}
\end{align}

The numbers $\{\theta_u\}$ are the conditional probabilites of the uniform distribution on leaves.
The $\{\delta_u\}$ values
are employed as ``noise'' added to the entropy calculation.
Such noise is a fundamental aspect for competitive analysis, and distinguishes it from the application
of mirror descent to regret minimization problems (see, e.g., \cite{BN12}).\footnote{One finds aspects of this ``mixing with the uniform distribution''
	in the bandits setting as well, but used for variance reduction, a seemingly very different purpose.}
The effect of these noise parameters appears ubiquitously in applications of the primal-dual method to competitive analysis (see \cite{BN07}),
and manifests itself as an additive term in the update rules (see \pref{eq:evo} below). Intuitively, it ensures that the conditional probability $\frac{x_u}{x_{\sfp(u)}}$ is updated fast enough even when it is close to $0$.

Finally, the numbers $\{\eta_u : u \in V\}$ are commonly referred to as ``learning rates'' in 
the study of online learning.  They represent the rate at which information is discounted in the resulting algorithm;
for MTS, this corresponds to the relative importance of costs arriving now vs. costs that arrived in the past.

\subsubsection{The dynamics}
\label{sec:dyn}

We will derive in \pref{sec:derivation} the following \emph{continuous time} evolution of the resulting mirror descent algorithm $\left(x(t) \in \K_T : t \in [0,\infty)\right)$ for a cost path $c\colon[0,\infty)\to \R_+^{\cL}$:
\begin{equation}\label{eq:evo}
\partial_t \left(\frac{x_u(t)}{x_{\sfp(u)}(t)}\right) = \frac{\eta_u}{w_u} \left(\frac{x_u(t)}{x_{\sfp(u)}(t)}+\delta_u\right)
\left(\beta_{\sfp(u)}(t) - \sum_{\ell \in \cL_u} \frac{x_{\ell}(t)}{x_u(t)} c_{\ell}(t)\right)
\end{equation}
Here, $\beta_{\sfp(u)}(t)$ is a Lagrangian multiplier that ensures conservation of conditional probability:
\[
\sum_{v \in \chi(\sfp(u))} \partial_t \left(\frac{x_v(t)}{x_{\sfp(u)}(t)}\right) = 0\,.
\]
One can see that the evolution is being driven by the expected instantaneous
cost incurred conditioned on the current state being in the subtree rooted at $u$.

One should interpret \eqref{eq:evo} only when $x(t)$ lies in the relative interior of $\K_T$. Otherwise,
the conditional probabilities are ill-defined.
One way to rectify this is to
prevent $x(t)$ from hitting the relative boundary of $\K_T$ at all.
It is possible to adaptively modify the cost functions by a suitably small perturbation so as to guarantee
this property and, at the same time, ensure that the total discrepancy between
the modified and true service cost
is a small additive constant.

Instead, we will follow a different approach, by extending the dynamics
to an analogous system of conditional probabilities
$\{q_u(t) : u \in V \setminus \{\rt\}\}$:
\begin{equation}\label{eq:star-dynamics}
  \partial_t q_u(t) = \frac{\eta_u}{w_u} \left(q_u(t)+\delta_u\right) \left(\beta_{\sfp(u)}(t) - \hat{c}_u(t) + \alpha_u(t)\right),
\end{equation}
where $q_u(t) = \frac{x_u(t)}{x_{\sfp(u)}(t)}$ whenever $x_{\sfp(u)}(t) > 0$, $\alpha_u(t)$ is a Lagrangian multiplier for the constraint $q_u(t)\ge 0$, and $\hat{c}_u(t)$ is the ``derived'' cost
in the subtree rooted at $u$:
\begin{align*}
\hat{c}_u(t) &\seteq \sum_{\ell \in \cL_u} q_{\ell\mid u}(t) c_{\ell}(t)\\
q_{\ell \mid u}(t) &\seteq \prod_{v \in \gamma_{u,\ell} \setminus \{u\}} q_v(t)\,,\nonumber
\end{align*}
where $\gamma_{u,\ell}$ is the unique simple $u$-$\ell$ path in $T$.

Stated this way, the mirror descent algorithm can be envisioned as
running a ``weighted star'' algorithm on the conditional
probabilities at every internal node of $T$, with the derived
costs at an internal node $u$ given by the average cost of the current strategy
for playing one unit of mass in the subtree rooted at $u$.

In the next section, we will implement and analyze a discretization of \eqref{eq:star-dynamics}
using Bregman projections.
Since our regularizer $\Phi$ and convex body $\K_T$ do not satisfy
the assumptions underlying the existence and uniqueness theorem of \cite{BCLLM18},
we need to construct a solution to \eqref{eq:star-dynamics}
and, indeed, taking the discretization parameter in our algorithm to zero,
one establishes a solution of bounded variation; see \pref{sec:continuous}.

The major benefit of the formulations \eqref{eq:evo} and \eqref{eq:star-dynamics}
is in motivating such an algorithm and prescribing the derived costs.
In \pref{sec:derivation}, we describe how these dynamics
can be predicted from the definition \eqref{eq:reg1}.

\section{The MTS algorithm}\label{sec:main}

We will first establish some generic machinery which, at this point, is not specific to MTS yet. Consider a convex polytope 
$\sfK_0\subseteq \R^n$, define $\sfK \seteq \sfK_0 \cap \R_+^n$, and assume that $\sfK$ is compact.
Suppose additionally that $\Phi : \cD \to \R$
is differentiable and strictly convex in an open neighborhood $\cD\supseteq \sfK$.

Let us write $\sfD_{\Phi}$ for the corresponding Bregman divergence
\[
   \sfD_{\Phi}(y \dmid x) \seteq \Phi(y) - \Phi(x) - \llangle \nabla \Phi(x), y-x\rrangle,
\]
which is non-negative due to convexity of $\Phi$.
Then for $x,y,z \in \K$, we have:
\begin{equation}\label{eq:s1}
   \sfD_{\Phi}(z \dmid y) - \sfD_{\Phi}(z \dmid x) = -\Phi(y)+\Phi(x)-\langle \nabla \Phi(y), z-y\rangle+\langle \nabla \Phi(x),z-x\rangle.
\end{equation}

For a vector $c \in \R^n$ and $x \in \K$, define the projection
\[
   \Pi_{\sfK}^c (x) \seteq \argmin \left\{ \sfD_{\Phi}(y \dmid x) + \langle c,y\rangle : y \in \sfK \right\}.
\]
Since $\K$ is compact and $\Phi$ is strictly convex, there is a unique minimizer $y^* \in \K$.

For $x \in \K$, recall the definition of the normal cone at $x$:
\[\sfN_{\sfK}(x) = \left\{ p \in \R^n : \langle p, y-x\rangle \leq 0 \textrm{ for all } y \in \K \right\}.\]
Given a representation of $\K$ by inequality constraints, $\K=\{x\in\R^n\colon Ax\le b\}$ for $A\in\R^{m\times n}$ and $b\in\R^n$, it holds
\begin{align*}
\sfN_{\sfK}(x) = \{A^Ty\colon y\ge0\text{ and }y^T(Ax-b)=0\}.
\end{align*}
The KKT conditions yield
\begin{equation}\label{eq:KKT}
   \nabla \Phi(y^*) = \nabla \Phi(x) - c - \lambda^*\,,
\end{equation}
where $\lambda^* \in \sfN_{\sfK}(y^*)$.
Since $\sfN_{\sfK}(y^*) = \sfN_{\K_0}(y^*) + \sfN_{\R_+^n}(y^*)$, we can
can decompose $\lambda^* = \beta - \alpha$ 
with $\beta \in \sfN_{\sfK_0}(y^*)$ and $- \alpha \in \sfN_{\R_+^n}(y^*)$.
In particular, we have $\alpha \geq 0$ and
$\alpha_i > 0 \implies y^*_i = 0$ for every $i=1,\ldots,n$.

Substituting this into \eqref{eq:s1} gives
\begin{align*}
   \sfD_{\Phi}(z \dmid y^*) - \sfD_{\Phi}(z \dmid x) &= -\Phi(y^*)+\Phi(x)+\langle \nabla \Phi(x),y^*-x\rangle +
   \langle c-\alpha + \beta, z-y^* \rangle \\
   &\leq -\sfD_{\Phi}(y^* \dmid x) + \langle c-\alpha, z-y^*\rangle,
\end{align*}
where the inequality comes from $\langle \beta, z-y^*\rangle \leq 0$ since $z \in \sfK$ and $\beta \in \sfN_{\sfK}(y^*)$.
We have proved the following.

\begin{lemma}\label{lem:bp}
   For any $x, z \in \sfK$, and $c \in \R^n$, let $y^* = \Pi_{\K}^c(x)$ and $\lambda^*$
   be as in \eqref{eq:KKT}.
   Then for any $\alpha \in -\sfN_{\R^+_n}(y^*)$ such that $\lambda^* + \alpha \in \sfN_{\K_0}(y^*)$,
   it holds that
   \[
      \sfD_{\Phi}(z \dmid y^*) - \sfD_{\Phi}(z \dmid x) \leq \langle c-\alpha, z-y^*\rangle.
   \]
\end{lemma}

\subsection{Iterative Bregman projections}
\label{sec:iter}

We describe now a discretization of the algorithm from the introduction. This discretization will mimic the continuous dynamics if the entries of each individual cost vector are small. We can achieve this by splitting each cost vector into several copies of scaled down versions of itself, as discussed in \pref{sec:alg}. In \pref{sec:continuous}, we will give a formal argument that this indeed yields a discretization of the continuous dynamics from the introduction.

Fix a tree $T$ and recall the definition of $\K_T$ from \eqref{eq:mtsK}.
Let $Q_{T}$ denote the collection of vectors $q \in \R_+^{V \setminus \{\rt\}}$ such that
for all $u \in V \setminus \cL$,
\[
	\sum_{v \in \chi(u) } q_v = 1.
\]
For $q \in Q_T$ and $u \in V \setminus \cL$, we use $q^{(u)} \in \R_+^{\chi(u)}$ to denote the vector
defined by $q^{(u)}_v \seteq q_v$ for $v \in \chi(u)$, and define
the corresponding probability simplex
$Q^{(u)}_T \seteq \{ q^{(u)} : q \in Q_T \}$.
We will use $\Delta : Q_T \to \K_T$ for the map which sends $q \in Q_T$ to the (unique) $x = \Delta(q) \in \K_T$ such that
\[
   x_v = x_u q_v \qquad \forall u \in V \setminus \cL, v \in \chi(u).
\]
Note that $q$ contains more information than $x$; the map $\Delta$ fails to be invertible whenever 
there is some $u \in V \setminus \cL$ with $x_u = 0$.

Fix $\kappa \ge 1$. On the open domain $\cD^{(u)}=(-\min_{v\in\chi(u)}\delta_v,\infty)^{\chi(u)}$, for $\delta_v$ as given in \eqref{eq:deltadef}, define the strictly convex function $\Phi^{(u)} : \cD^{(u)} \to \R$ by
\[
   \Phi^{(u)}(p) \seteq \frac{1}{\kappa} \sum_{v \in \chi(u)} \frac{w_v}{\eta_v} \left(p_v + \delta_v\right) \log \left(p_v + \delta_v\right).
\]
Denote the corresponding Bregman divergence on $Q^{(u)}_T$ by
\[
   \sfD^{(u)}\!\left(p \dmid p' \right) = \frac{1}{\kappa} \sum_{v \in \chi(u)} \frac{w_v}{\eta_v} \left[\left(p_v + \delta_v\right) \log \frac{p_v+\delta_v}{p'_v+\delta_v}
+ p'_v-p_v\right].
\]

We now define an algorithm that takes a point $q \in Q_T$ and a cost vector $c \in \R_+^{\cL}$ and outputs
a point $p = \cA(q,c)\in Q_T$.
Fix $\langle u_1,u_2,\ldots, u_N\rangle$ a topological ordering of $V \setminus \cL$
such that every child in $T$ occurs before its parent.
We define $p$ inductively as follows.
Let $\hat{c}_{\ell} \seteq c_{\ell}$ for $\ell \in \cL$.
For every $j=1,2,\ldots,N$:
\begin{align}
   \hat{c}^{(u_j)}_v &\seteq \hat{c}_v \qquad \forall v \in \chi(u_j) \label{eq:hatc}\\
   p^{(u_j)} &\seteq \argmin \left\{ \sfD^{(u_j)}\!\left(p \dmid q^{(u_j)}\right) + \llangle p, \hat{c}^{(u_j)} \rrangle \bigmid 
   p \in Q^{(u_j)}_T\right\} \label{eq:locp} \\
   \hat{c}_{u_j} &\seteq \sum_{v \in \chi(u_j)} p^{(u_j)}_v\,\hat{c}_v \label{eq:hatc2}
\end{align}
Let $\alpha^{(u_j)}$ be
the vector of Lagrange multipliers corresponding to the nonnegativity constraints in \eqref{eq:locp}
(recall \pref{lem:bp}).
One should note that in this setting (a probability simplex), the nonnegativity multipliers are unique
and thus well-defined.

We denote $\alpha = \alpha^{q,c} \in \R_+^{V}$ as the vector given by $\alpha_v \seteq \alpha^{(\sfp(v))}_v$ for $v \neq \rt$
and $\alpha_{\rt} \seteq 0$.
Recall the complementary slackness conditions:
\begin{equation}\label{eq:posmult}
   \alpha_v > 0 \implies p_v = 0.
\end{equation}
For $v \in \chi(u)$, calculate
\[
   \left(\nabla \Phi^{(u)}(p)\right)_v = \frac{1}{\kappa} \frac{w_v}{\eta_v} \left(1+\log(p_v+\delta_v)\right).
\]
Then using \eqref{eq:KKT}, we can write the algorithm as follows:
\begin{quote}
\begin{tabbing}
   For \=$j=1,2,\ldots,N$: \\[0.1cm]
   \>For \=$v \in \chi(u_j)$: \\[0.1cm]
   \>\>$p^{(u_j)}_v \seteq (q^{(u_j)}_v+\delta_v) \exp\left(\kappa \frac{\eta_v}{w_v} \left(\beta_{u_j} - (\hat{c}_v - \alpha_v)\right)\right)- \delta_v$, \\[0.1cm]
   \>$\hat{c}_{u_j} \seteq \sum_{v \in \chi(u_j)} p_v^{(u_j)} \hat{c}_v$.
\end{tabbing}
\end{quote}
where $\beta_{u_j}\ge 0$ is the multiplier for the constraint $\sum_{v \in \chi(u_j)} q^{(u_j)}_v \geq 1$. There is no multiplier for the constraint $\sum_{v \in \chi(u_j)} q_v^{(u_j)} \leq 1$ because this constraint will be satisfied automatically and is therefore not needed in \eqref{eq:locp}: If it were violated, decreasing some $p_v$ with $p_v>q_v^{(u_j)}$ would yield a strictly better solution to the minimization problem \eqref{eq:locp}.

\subsection{The global divergence}

For $z \in \K_T$ and $q \in Q_T$, 
define the global divergence function
\[
   \tD(z \dmid q) \seteq \frac1{\kappa} \sum_{u \notin \cL} \sum_{v \in \chi(u)} \frac{w_v}{\eta_v} \left[ \left(z_v + \delta_v z_u\right) \log \left(\frac{\frac{z_v}{z_u}+\delta_v}{q_v+\delta_v}\right)
      + z_u q_v - z_v
\right],
\]
with the convention that $0 \log \left(\frac{0}{0}+\delta_v\right) = \lim_{\e \to 0} \e \log \left(\frac{0}{\e}+\delta_v\right) = 0$. We remark that $\tD$ is the Bregman divergence associated to the regularizer \pref{eq:reg1} (divided by $\kappa$) when $\frac{x_v}{x_{u}}$ is replaced by $q_v$. We will use $\tD$ as a potential function to prove inequality \pref{eq:finesc}. Here, $z$ denotes the configuration of the offline algorithm. Note that while the online configuration is encoded by $q\in Q_T$, we still use the polytope $\sfK_T$ to encode the offline configuration. This will be more convenient when expressing the offline movement cost.

The next lemma shows that when the offline algorithm moves, the change in potential is bounded by $O(1/\kappa)$ times the offline movement cost.

\begin{lemma}\label{lem:lipschitz}
      It holds that for any $q \in Q_T$ and $z,z' \in \sfK_T$,
      \[
         \left|\tD(z \dmid q) - \tD(z' \dmid q)\right| \leq \frac{1}{\kappa} \left(2+\frac{4}{\tau}\right) \left\|z-z'\right\|_{\ell_1(w)}.
\]
\end{lemma}

\begin{proof}
   Consider a differentiable map $z\colon[0,1]\to \R_{++}^V$ such that $\sum_{v\in\chi(u)}z_v(t)\le z_u(t)$ for each $t$ and $u\notin\cL$.
   It suffices to show that for each $t$ and every fixed $q \in Q_T$,
   \[
      \kappa \left|\partial_t \tD(z(t) \dmid q)\right| \leq \left(2+\frac{4}{\tau}\right) \left\|z'(t)\right\|_{\ell_1(w)}.
   \]
   Moreover, it suffices to address the case when there is at most one $u \in V$ with $z'_u(t) \neq 0$.

   A direct calculation gives
   \begin{align}
      \kappa \partial_t \tD(z(t) \dmid q) &= \frac{w_u}{\eta_u} z'_u(t) \log \left(\frac{z_u(t)/z_{\sfp(u)}(t)+\delta_u}{q_u+\delta_u}\right) \nonumber \\
                                   &\qquad 
                           + \sum_{v \in \chi(u)} \frac{w_v}{\eta_v} 
                                                   \left[\delta_v z'_u(t) \log \left(\frac{z_v(t)/z_u(t)+\delta_v}{q_v+\delta_v}\right)
                                                         + z'_u(t) \left(q_v - \frac{z_v(t)}{z_u(t)}\right)
                                                      \right].\label{eq:tD}
   \end{align}
   Let us now use definitions \eqref{eq:etadef} and \eqref{eq:deltadef}
   to observe that
   \[
      \frac{1}{\eta_v} \left|\log \frac{p_v+\delta_v}{q_v+\delta_v}\right| \leq \frac{1}{\eta_v} \log \frac{1+\delta_v}{\delta_v} \leq 2.
   \]
   Using this in \eqref{eq:tD} yields
   \[
      \kappa \left|\partial_t \tD(z(t) \dmid q)\right| \leq w_u |z_u'(t)| \left(2 + \frac{1}{\tau} \sum_{v \in \chi(u)}
      \left(2\delta_v + \left|q_v-\frac{z_v(t)}{z_u(t)}\right|\right)\right)
      \leq w_u |z_u'(t)| \left(2+\frac{4}{\tau}\right),
   \]
   where the last inequality uses $\sum_{v \in \chi(u)} \delta_v \leq \sum_{v \in \chi(u)} \theta_v \leq 1$ and $\sum_{v\in\chi(u)}z_v(t)\le z_u(t)$.
\end{proof}

We will sometimes implicitly restrict vectors $x \in \R^{V}$ to the subspace spanned by $\{ e_{\ell} : \ell \in \cL \}$.
In this case, we employ the notation
\[
   \langle x,y \rangle_{\cL} \seteq \sum_{\ell \in \cL} x_{\ell} y_{\ell},
\]
when either vector lies in $\R^V$ or $\R^{\cL}$.

According to the following lemma, the change in potential due to movement of the online algorithm is bounded by the difference in service cost between the offline and online algorithm.
\begin{lemma}\label{lem:sc}
   For any cost vector $c \in \R_+^{\cL}$, $z \in \K_T$, and $q \in Q_T$, it holds that if $p = \cA(q,c)$, then
   \[
      \tD(z \dmid p) - \tD(z \dmid q) \leq \llangle c, z - \Delta(p)\rrangle_{\cL}.
   \]
\end{lemma}

\begin{proof}
   Fix $q \in Q_T$ and $c \in \R_+^{\cL}$.
   Let $\alpha = \alpha^{q,c}$ denote the vector of multipliers
   defined in \pref{sec:iter}.
For $u \in V \setminus \cL$ with $z_u > 0$, define $z^{(u)} \in Q^{(u)}_T$ by
\[
   z^{(u)}_v \seteq \frac{z_v}{z_u}. 
\]
Then \pref{lem:bp} gives
\[
   \sfD^{(u)}\left(z^{(u)} \dmid p^{(u)}\right) - \sfD^{(u)}\left(z^{(u)} \dmid q^{(u)}\right) \leq \llangle \hat{c}^{(u)} - \alpha^{(u)}, 
   z^{(u)} - p^{(u)}\rrangle_{\chi(u)},
\]
where we use $\langle \cdot,\cdot\rangle_{\chi(u)}$ for the standard inner product on $\R^{\chi(u)}$.
Multiplying by $z_u$ and summing yields
\begin{align*}
   \tD(z \dmid p) - \tD(z \dmid q) &\leq \sum_{u \notin \cL} z_u \llangle \hat{c}^{(u)} - \alpha^{(u)}, z^{(u)} - p^{(u)} \rrangle_{\chi(u)} \\
                                         &= \sum_{u \notin \cL} \sum_{v \in \chi(u)} (\hat{c}^{(u)}_v -\alpha^{(u)}_v) z_v - \sum_{u \notin \cL} z_u \sum_{v \in \chi(u)} (\hat{c}^{(u)}_v-\alpha^{(u)}_v) p_v\,.
\end{align*}
Note that from \eqref{eq:posmult}, the latter expression is
\[
  \sum_{u \notin \cL} z_u \sum_{v \in \chi(u)} \hat{c}_v^{(u)} p_v
  \stackrel{\eqref{eq:hatc2}}{=}\sum_{u \notin \cL} z_u \hat{c}_u.
\]
Noting that $\hat{c}_{\vvr} = \sum_{\ell \in \cL} \Delta(p)_{\ell} c_{\ell}$,  this gives
\[
   \tD(z \dmid p) - \tD(z \dmid q) \leq \sum_{u \neq \rt} (\hat{c}_u - \alpha_u) z_u - \sum_{u \notin \cL} z_u \hat{c}_u
   \leq \llangle c, z-\Delta(p)\rrangle_{\cL}. \qedhere
\]
\end{proof}

\subsection{Algorithm and competitive analysis}
\label{sec:alg}

Let us now outline the proof of inequality \pref{eq:finemc}.
First, we perform a standard reduction that allows us to
bound only the ``positive'' movement costs when the algorithm
moves from $x$ to $y$.
Its proof is straightforward.

\begin{lemma}\label{lem:height}
	For $x,y\in \sfK_T$ it holds that
	\[
	\|x-y\|_{\ell_1(w)} = 2 \left\|(x-y)_{+}\right\|_{\ell_1(w)} + [\psi(y)-\psi(x)],
	\]
	where $\psi(x) \seteq \sum_{u \neq \vvr} w_u x_u$ for $x \in \sfK_T$.
\end{lemma}

We now state the key technical lemma
which controls the positive movement cost by
the service cost.  To this end, we employ an auxiliary potential function $\Psi : Q_T \to \R$ defined by
\begin{align*}
\Psi_u(q) &\seteq - \Delta(q)_u\,\sfD^{(u)}\!\left(\theta^{(u)} \dmid q^{(u)}\right) \\
\Psi(q) &\seteq \sum_{u \notin \cL} \Psi_u(q).
\end{align*}
Intuitively, $\Psi(q)$ is a measure of difference between the online configuration $q$ and the uniform distribution over leaves (whose conditional probabilities are given by $\theta$).

Let us give a brief explanation of the need for $\Psi$.
Our addition of ``noise'' to the multiscale conditional entropy
is to achieve the smoothness property established in \pref{lem:lipschitz}.
But this has the adverse effect of increasing the movement cost
of the algorithm, as one can see from the $\delta_u$ term in \pref{eq:evo}.
This additional movement cannot be easily charged against the service cost in
the regime where the noise term is dominant:  $\frac{x_u(t)}{x_{\sfp(u)}(t)} \ll \delta_u$.
On the other hand, this additional movement has the effect of further decreasing $\frac{x_u(t)}{x_{\sfp(u)}(t)}$,
which drives the conditional probabilities at $\sfp(u)$ away from the uniform distribution,
decreasing $\Psi$.  A formal statement appears later in \pref{lem:crucial}.

For the next two results, take any $q \in Q_T$ and cost $c \in \R_+^{\cL}$, and denote $p=\cA(q,c), x =\Delta(q),y=\Delta(p)$.

\begin{lemma}[Movement analysis]
	\label{lem:movement}
	It holds that
	\[
	\frac{\tau-3}{\kappa \tau} \left\|\left(x-y\right)_+\right\|_{\ell_1(w)} \leq (2 \depth_T + \log n) \langle c, x\rangle_{\cL} +
	\left[\Psi(q)-\Psi(p)\right].
	\]
\end{lemma}

This lemma will be proved in \pref{sec:movement-analysis}.
Let us first see that it can be used to establish bounds on the competitive ratio.
Define $w_{\min} \seteq \min \{ w_{\ell} : \ell \in \cL \}$
and
\[
\eps_T \seteq \frac{w_{\min}}{2 (2\depth_T+\log n)} \frac{\tau-3}{\tau\kappa}.
\]

\begin{theorem}\label{thm:analysis}
	For any $z \in \sfK_T$:
	\begin{align}
	\langle c, y\rangle_{\cL} &\leq \langle c,z\rangle_{\cL} + \left[\tD(z \dmid q) - \tD(z \dmid p)\right] \label{eq:sc} \\
	\kappa^{-1} \left\|x-y\right\|_{\ell_1(w)} &\leq [\psi(y)-\psi(x)] + \frac{2 \tau}{\tau-3} \left([\Psi(q)-\Psi(p)] + (2\depth_T + \log n) \langle c, x\rangle_{\cL}\right)\label{eq:mvmtx}
	\end{align}
Moreover, if $\|c\|_{\infty} \leq \e_T$, then
\begin{equation}\label{eq:mvmty}
   \kappa^{-1} \left\|x-y\right\|_{\ell_1(w)} \leq [\psi(y)-\psi(x)] + \frac{4 \tau}{\tau-3} \left( [\Psi(q)-\Psi(p)] + (2\depth_T + \log n) \langle c, y\rangle_{\cL}\right).
\end{equation}
\end{theorem}

\begin{proof}
   The bound \eqref{eq:sc} follows from \pref{lem:sc}, and \eqref{eq:mvmtx}
   follows from \pref{lem:movement} and \pref{lem:height}.
   To see that \eqref{eq:mvmty} follows from \eqref{eq:mvmtx} and \pref{lem:movement}, use the fact that
   \[
      \langle c,x \rangle_{\cL} \leq \langle c,y\rangle_{\cL} + \frac{\|c\|_{\infty}}{w_{\min}} \left\|\left(x-y\right)_+\right\|_{\ell_1(w)}.\qedhere
   \]
\end{proof}

In light of \pref{thm:analysis}, we can respond to a cost function $c \in \R_+^{\cL}$ by splitting it into $M$ pieces
$c_1,c_2,\ldots,c_M$ where $M = \lceil \|c\|_{\infty}/\e_T\rceil$.
Now define $q_i \seteq \cA(q_{i-1}, c/M)$, $q_0:=q$ and $\bar{\cA}(q,c) \seteq q_M$.

\begin{theorem}
   \label{thm:main-technical}
   Fix $\tau \geq 4$.
   Consider the algorithm that begins in some configuration $q_0 \in Q_T$.
   If $c_t \in \R_+^{\cL}$ is the cost function that arrives at time $t$,
   denote $q_t \seteq \bar{\cA}(q_{t-1},c_t)$.  Then the sequence
   $\llangle \Delta(q_0), \Delta(q_1), \ldots\rrangle$ is an online algorithm
   that is $(1,O(1/\kappa), O(\kappa(\depth_T + \log n)))$-finely competitive.
\end{theorem}

We prove this momentarily.
The following fact is well-known and, in conjunction with the preceding theorem,
yields the validity of \pref{thm:refined} and \pref{thm:super}.

\begin{lemma}[cf. \cite{BBMN15}]
   If $(\cL,d_T)$ is an HST metric, then there is another weighted tree $T'$ with leaf set $\cL$
   such that
   \begin{enumerate}
      \item $(\cL,d_{T'})$ is a $7$-HST metric.
      \item  $\depth_{T'} \leq \log_2 |\cL|$
      \item All the leaves of $T'$ have depth $\depth_{T'}$.
      \item $d_T(\ell,\ell') \leq d_{T'}(\ell,\ell') \leq O(d_T(\ell,\ell'))$ for all $\ell,\ell' \in \cL$.
   \end{enumerate}
\end{lemma}
\begin{proof}[Proof sketch]
   Replace every weight $w_v$ in $T$ with $\hat{w}_v \seteq 7^{\lceil \log_7 w_v\rceil}$ and iteratively contract every edge $(p(u),u)$
   with $\hat{w}_{p(u)}=\hat{w}_u$ and $u \notin \cL$.
   The resulting weighted tree $T_1$ is a $7$-HST by construction.

   Now iteratively contract every edge $(p(u),u)$ in $T_1$ for which $|\cL^{T_1}_{u}| > \frac12 |\cL^{T_1}_{p(u)}|$.
   The resulting tree $T'$ has depth $\depth_{T'} \leq \log_2 |\cL|$.
   Finally, one can achieve property (3) by increasing the depth of every root-leaf path to $\depth_{T'}$
   using vertex weights that decrease by a factor of $7$ along the path.
\end{proof}

\begin{proof}[Proof of \pref{thm:main-technical}]
   Consider a sequence $\llangle c_t : t \geq 1\rrangle$ of cost functions.
   By splitting the costs into smaller pieces, we may assume that $\|c_t\|_{\infty} \leq \e_T$ for all $t \geq 1$.

   Let $\{z_t^*\}$ denote some offline algorithm with $z_0^* = \Delta(q_0)$, and let $\{x_t = \Delta(q_t)\}$ denote our
   online algorithm.  Then using $\tD(z_{0}^* \dmid x_0) = 0$ along with \eqref{eq:sc} and \pref{lem:lipschitz}
   yields, for any time $t_1 \geq 1$,
   \begin{align*}
      \sum_{t=1}^{t_1} \langle c_t, x_t\rangle_{\cL} &\leq 
      \sum_{t=1}^{t_1} \langle c_t,z_t^*\rangle_{\cL} -
      \tD(z_{t_1}^* \dmid q_{t_1}) + O(1/\kappa) \sum_{t=1}^{t_1} \|z_t^*-z_{t-1}^*\|_{\ell_1(w)} \\
      &\leq \sum_{t=1}^{t_1} \langle c_t,z_t^*\rangle_{\cL}
         + O(1/\kappa) \sum_{t=1}^{t_1} \|z_t^*-z_{t-1}^*\|_{\ell_1(w)},
\end{align*}
where we have used $\tD(z \dmid q) \geq 0$ for all $z \in \K_T$ and $q \in Q_T$.
This verifies \eqref{eq:finesc} with $\alpha_0=1$, $\alpha_1 = O(1/\kappa)$, and $\beta = 0$.
Moreover, \eqref{eq:mvmty} gives
\[
   \frac{1}{\kappa} \sum_{t=1}^{t_1} \|x_t-x_{t-1}\|_{\ell_1(w)} \leq 
   \left[\psi(x_{t_1})-\psi(x_0)\right]+\frac{4\tau}{\tau-3} \left[\Psi(q_0)-\Psi(q_{t_1})\right] + \left(2\depth_T + \log n\right) 
   \sum_{t=1}^{t_1} \langle c_t,x_t\rangle_{\cL},
\]
verifying \eqref{eq:finemc} with $\alpha_1 \leq O(\kappa(\depth_T + \log n))$ and $\beta' \leq O(\kappa \max_{v \neq \rt} w_v)$ (see \pref{lem:dmax} below).
\end{proof}

\subsection{Movement analysis}\label{sec:movement-analysis}

It remains to prove \pref{lem:movement}.
The KKT conditions (cf. \eqref{eq:KKT}) give:  For every $v \in \chi(u)$,
\begin{equation}\label{eq:kkt2}
   \frac{1}{\kappa} \frac{w_v}{\eta_v} \log \left(\frac{p_v+\delta_v}{q_v+\delta_v}\right) = \beta_u - \hat{c}_v + \alpha_v\,,
\end{equation}
where $\beta_u \geq 0$ is the multiplier corresponding to the constraint $\sum_{v \in \chi(u)} q_v \geq 1$.

\begin{lemma}\label{lem:alphas}
   It holds that $\alpha_v \leq \hat{c}_v$ for all $v \in V \setminus \{\rt\}$.
\end{lemma}

\begin{proof}
   Note that $\hat{c}_v \geq 0$ by construction.  Thus if $\alpha_v = 0$, we are done.
   Otherwise, by complementary slackness, it must be that $p_v = 0$, and therefore
   $\log(\frac{p_v+\delta_v}{q_v+\delta_v}) \leq 0$.  Since $\beta_{\sfp(v)} \geq 0$, \eqref{eq:kkt2} implies that
   $\alpha_v \leq \hat{c}_v$.
\end{proof}

Define
   $\sigma_v \seteq \log \left(\frac{p_v+\delta_v}{q_v+\delta_v}\right)$
so that
\begin{align}
   q_v - p_v &= (q_v + \delta_v) (1-e^{\sigma_v}). \label{eq:qp}
\end{align}
Recall that for $v \in \chi(u)$, we have $x_v = q_v x_u$ and $y_v = p_v y_u$, thus
\[
   x_v - y_v = x_u (q_v - p_v) + p_v (x_u - y_u) = (x_v+\delta_v x_u) (1-e^{\sigma_v}) + p_v (x_u - y_u).
\]
In particular,
\begin{align*}
   w_v \left(x_v-y_v\right)_+ &\leq w_v (x_v+\delta_v x_u) (1-e^{\sigma_v})_+ + w_v p_v \left(x_u-y_u\right)_+ \\
                              &\leq w_v (x_v+\delta_v x_u) (1-e^{\sigma_v})_+ + \frac{w_u}{\tau} p_v \left(x_u-y_u\right)_+.
\end{align*}
Using $\sum_{v \in \chi(u)} p_v = 1$ and
summing over all vertices yields
\[
   \sum_{v \neq \rt} w_v \left(x_v-y_v\right)_+ \leq \sum_{v \neq \rt} w_v (x_v+\delta_v x_{\sfp(v)}) (1-e^{\sigma_v})_{+} + \frac{1}{\tau} \sum_{v \neq \rt} w_v \left(x_v-y_v\right)_+\,,
\]
hence
\begin{align}
   \sum_{v \neq \rt} w_v \left(x_v-y_v\right)_+ &\leq \frac{\tau}{\tau-1} \sum_{v \neq \rt} w_v (x_v+\delta_v x_{\sfp(v)}) (1-e^{\sigma_v})_+ \nonumber \\ 
                                                &\leq \frac{\tau}{\tau-1} \sum_{v \neq \rt} w_v (x_v+\delta_v x_{\sfp(v)}) \left(\sigma_v\right)_{-} \nonumber \\
                                                &\leq\frac{\kappa \tau}{\tau-1} \left(\sum_{v \neq \rt} \eta_v x_v \hat{c}_v + \sum_{u \notin \cL} x_u \sum_{v \in \chi(u)} \theta_v (\hat{c}_v-\alpha_v)\right),
   \label{eq:mvmt0}
\end{align}
where the last line uses
\pref{lem:alphas} and \eqref{eq:kkt2}, to bound
$w_v (\sigma_v)_{-} \leq \kappa \eta_v\left(\hat{c}_v - \alpha_v\right)$.

 Note that
\begin{align}\label{eq:mv2}
   \sum_{v \neq \rt} \eta_v x_v \hat{c}_v \leq \sum_{\ell \in \cL} c_{\ell} x_{\ell} \sum_{v \in \gamma_{\rt,\ell} \setminus \{\rt\}} \eta_v
   \leq (\depth_T + \log n) \llangle c,x\rrangle,
\end{align}
since for any $\ell \in \cL$, it holds that
\begin{align*}
   \sum_{v \in \gamma_{\rt,\ell} \setminus \{\rt\}} \eta_v = \depth_T(\ell) + \sum_{v \in \gamma_{\rt,\ell} \setminus \{\rt\}} \log \frac{|\cL_{\sfp(v)}|}{|\cL_v|} =
\depth_T(\ell) + \log n,
\end{align*}
where $\depth_T(\ell)$ is the combinatorial depth of $\ell$.

The second sum in \eqref{eq:mvmt0} can be interpreted as the service cost of hybrid configurations of $q$ and $\theta$: While $\sum_{v\in\chi(u)}x_v\hat{c}_v$ is the service cost of $x$ in $\cL_u$, the term $x_u\sum_{v\in\chi(u)}\theta_v\hat{c}_v$ is the service cost in $\cL_u$ of the modification of $x$ whose conditional probabilities at the children of $u$ are given by $\theta^{(u)}$ rather than $q^{(u)}$. To bound this hybrid service cost, we will employ the auxiliary potential $\Psi$.

\subsubsection{The hybrid cost}
\label{sec:hybrid}

We require the following elementary estimate.

\begin{lemma}\label{lem:dmax}
   For $u \notin \cL$ it holds that
   \begin{equation*}\label{eq:dmax}
      \max \left\{ \sfD^{(u)}(r \dmid p) : r,p \in Q^{(u)}_T \right\} \leq \frac{2}{\kappa} \frac{w_u}{\tau}.
   \end{equation*}
\end{lemma}

\begin{proof}
   Define $\phi_v : (-\delta_v,\infty) \to \R$ by
   \[
      \phi_v(p) \seteq \frac{1}{\eta_v} (p_v+\delta_v) \log (p_v+\delta_v),
   \]
   and let
   \[
      \sfD_{\phi_v}(q_v \dmid p_v) = \frac{1}{\eta_v} \left[(q_v+\delta_v) \log \frac{q_v+\delta_v}{p_v+\delta_v} + (p_v-q_v)\right]
   \]
   denote the corresponding Bregman divergence.
   Then for $q_v,p_v \geq 0$, it holds that $\sfD_{\phi_v}(q_v \dmid p_v) \geq 0$ since $\phi_v$ is convex on $\R_+$.
   Employing the $\tau$-HST property of $T$, this implies that
   \[
      \sfD^{(u)}(r \dmid p) = \frac{1}{\kappa} \sum_{v \in \chi(u)} w_v \sfD_{\phi_v}(r_v \dmid p_v) \leq
      \frac{w_u}{\kappa \tau} \sum_{v \in \chi(u)} \sfD_{\phi_v}(r_v \dmid p_v).
   \]

   Define $F : Q^{(u)}_T \times Q^{(u)}_T \to \R_+$ by $F(r,p) \seteq\sum_{v \in \chi(u)} \sfD_{\phi_v}(r_v \dmid p_v)$.
   The map $r \mapsto F(r,p)$ is convex in general (for any Bregman divergence).
   The map $p \mapsto F(r,p)$ is convex as well, as this holds for each map $p_v \mapsto \sfD_{\phi_v}(q_v \dmid p_v)$
   since $-\log(x)$ is convex on $\R_{++}$.
   Since the maximum of a convex function on the a polytope is achieved at an extreme point, we have
   \begin{align*}
      \max \left\{ F(r,p) :r,p \in Q^{(u)}_T \right\} &\leq \max_{\substack{v,v' \in \chi(u) \\ v \neq v'}} 
      \left[
      \frac{1}{\eta_v} \left((1+\delta_v) \log \frac{1+\delta_v}{\delta_v} -1\right)
   + \frac{1}{\eta_{v'}} \left(\delta_{v'} \log \frac{\delta_{v'}}{1+\delta_{v'}} +1\right)\right]
      \\
      &\leq 2.\qedhere
   \end{align*}
\end{proof}

The next lemma is crucial: It relates the service cost (with respect to the reduced cost $\hat{c}-\alpha$) of the hybrid configurations to the service cost of the actual configuration and the movement cost.
\begin{lemma}\label{lem:crucial}
   For any $u \notin \cL$, it holds that
   \begin{equation}\label{eq:hybrid}
      \Psi_u(p) - \Psi_u(q) \leq \frac{2}{\kappa} \frac{w_u}{\tau} \left(x_u-y_u\right)_+ + 
      \sum_{v \in \chi(u)} (\hat{c}_v-\alpha_v) \left[x_v - \theta_v x_u\right].\qedhere
   \end{equation}
\end{lemma}

\begin{proof}
Write
\begin{align*}
   \Psi_u(p) - \Psi_u(q) &= 
    x_u \sfD^{(u)}\left(\theta^{(u)} \dmid q^{(u)}\right) 
   -y_u \sfD^{(u)}\left(\theta^{(u)} \dmid p^{(u)}\right)  \\
   &= (x_u-y_u) \sfD^{(u)}(\theta^{(u)} \dmid p^{(u)}) + x_u \left[\sfD^{(u)}(\theta^{(u)} \dmid q^{(u)})-\sfD^{(u)}(\theta^{(u)} \dmid p^{(u)})\right].
\end{align*}
Using \pref{lem:dmax}, the first term is bounded by $\frac{2}{\kappa} \frac{w_u}{\tau} (x_u-y_u)_+ $.

Let us now bound the second term.
Using $1+t \leq e^t$, we have
\begin{align*}
   \kappa x_u \left[\sfD^{(u)}(\theta^{(u)} \dmid q^{(u)})-\sfD^{(u)}(\theta^{(u)} \dmid p^{(u)})\right] &=
    x_u \sum_{v \in \chi(u)} \frac{w_v}{\eta_v} \left[(\theta_v+\delta_v) \log \frac{p_v+\delta_v}{q_v+\delta_v}+q_v-p_v\right] \\
   \  &\stackrel{\mathclap{\eqref{eq:qp}}}{=}\ 
    x_u \sum_{v \in \chi(u)} \frac{w_v}{\eta_v} \left[(\theta_v+\delta_v) \sigma_v + (q_v+\delta_v) (1-e^{\sigma_v})\right] \\
   &\leq\ 
   x_u \sum_{v \in \chi(u)} \frac{w_v}{\eta_v} \sigma_v (\theta_v-q_v) \\
   &= \sum_{v \in \chi(u)} \frac{w_v}{\eta_v} \sigma_v \left[\theta_v x_u - x_v\right].
\end{align*}
To finish the proof, observe that from \eqref{eq:kkt2},
\[
   \sum_{v \in \chi(u)} \frac{w_v}{\eta_v} \sigma_v \left[\theta_v x_u - x_v\right] =
\kappa \sum_{v \in \chi(u)} (\beta_u - \hat{c}_v + \alpha_v) \left[\theta_v x_u - x_v\right] =
\kappa \sum_{v \in \chi(u)} (\alpha_v - \hat{c}_v) \left[\theta_v x_u - x_v\right],
\]
where the last equality uses $\sum_{v \in \chi(u)} x_v = x_u$ and $\sum_{v \in \chi(u)} \theta_v = 1$ (from \eqref{eq:thetadef}).
\end{proof}

Using the lemma gives
\begin{align*}
   \sum_{u \notin \cL} x_u \sum_{v \in \chi(u)} \theta_v (\hat{c}_v-\alpha_v) 
   &\stackrel{\mathclap{\eqref{eq:hybrid}}}{\leq}\ 
   [\Psi(q)-\Psi(p)] + \frac{2}{\kappa \tau} \left\|\left(\Delta(q)-\Delta(p)\right)_{+}\right\|_{\ell_1(w)}
   + \sum_{v \neq \rt} \hat{c}_v x_v \\
   &\leq\ 
   [\Psi(q)-\Psi(p)] + \frac{2}{\kappa \tau} \left\|\left(\Delta(q)-\Delta(p)\right)_{+}\right\|_{\ell_1(w)}
   + \depth_T \langle c,x\rangle_{\cL}.
\end{align*}

Combining this inequality with \eqref{eq:mvmt0} and \eqref{eq:mv2} gives
\begin{equation}\label{eq:mvmt3}
   \kappa^{-1} \left\|\left(x-y\right)_+\right\|_{\ell_1(w)} 
\leq \frac{\tau}{\tau-1}\left[
\left(2\depth_T + \log n\right) \langle c,x\rangle_{\cL} + \left(\Psi(q)-\Psi(p)\right) 
   + \frac{2}{\kappa\tau} \left\|\left(x-y\right)_+\right\|_{\ell_1(w)}\right],
\end{equation}
completing the verification of \pref{lem:movement}.

\section{Derivation of the dynamics and derived costs}
\label{sec:derivation}

\newcommand{\Nperp}{N_{\K}^{\perp}}
\newcommand{\Dperp}{N_{\cD}^{\perp}}
\newcommand{\Ninv}{V_{\K}}
\newcommand{\rel}{\mathrm{rint}}
\newcommand{\ri}{\mathrm{ri}}

For the sake of motivating the dynamics \eqref{eq:evo}, we review the
continuous-time mirror descent framework of \cite{BCLLM18}.
Suppose that $\K \subseteq \R^N$ is a convex set.
We recall again the definition of the {\em normal cone to 
$\K$ at $x \in \K$} which is given by
\begin{alignat*}{3}
   N_{\K}(x) &\seteq (\K-x)^{\circ} &&= \left\{ p \in \R^N : \llangle p,y-x\rrangle \leq 0 \textrm{ for all } y \in \K\right\}.
\end{alignat*}

Suppose additionally that $\Phi : \cD \to \R$
is $\cC^2$ and strictly convex on an open neighborhood $\cD \supseteq \K$ so that
the Hessian $\nabla^2 \Phi(x)$ is well-defined and positive definite on $\cD$.
Given a control function $F : [0,\infty) \times \K \to \R^N$ and an initial point $x_0 \in \K$, we will be concerned
with absolutely continuous solutions $x : [0,\infty) \to \K$ to the differential inclusion
\begin{align*}
   x(0) &= x_0, \\
   \nabla^2 \Phi(x(t)) x'(t) &\in F(t,x(t)) - N_{\K}(x(t))\,.
\end{align*}
In other words, a trajectory that satisfies $x(0) = x_0$ and
for almost every $t \geq 0$:
\begin{equation}\label{eq:mirror}
   x'(t) = \nabla^2 \Phi(x(t))^{-1} \left(F(t,x(t)) - \gamma(t)\right),
\end{equation}
with $\gamma(t) \in N_{\K}(x(t))$.

Under suitably strong conditions on $\Phi$ and $F$, there is a unique absolutely continuous
solution to \eqref{eq:mirror} \cite{BCLLM18}.
In our setup, these conditions are actually {\em not satisfied} unless we prevent the path $x$ from hitting the relative boundary of $\K$.
Nevertheless, the formal calculation is elucidating and motivates the algorithm of \pref{sec:main}. For simplicity, we assume $\kappa:=1$ in this section.

\subsection{Hessian computation}

Let us take $\Phi$ as in \eqref{eq:reg1}
and calculate $\nabla^2 \Phi(x)$ for $x \in \R_{++}^{V}$.
Fix $u \neq \rt$. 
Then we have
\begin{align}
\, \partial_u\Phi(x) &= 
\frac{w_u}{\eta_u}
\left(\log\left(\frac{x_u}{x_{\sfp(u)}}+\delta_u\right) +1\right)
+ \sum_{v \in \chi(u)}\frac{w_v}{\eta_v}\left(\delta_v\log\left(\frac{x_v}{x_u}+\delta_v\right) - \frac{x_v}{x_u}\right). \label{eq:grad-phi}
\end{align}
Moreover, $\partial_{uv} \Phi(x) = 0$ unless $u=v$, $u \in \chi(v)$, or $v \in \chi(u)$, and in this case,
\begin{align*}
   \, \partial_{uu}\Phi(x) &= \frac{w_u}{\eta_u(x_u+\delta_u x_{\sfp(u)})} + \sum_{v\in \chi(u)} \left(\frac{x_v}{x_u}\right)^2\frac{w_v}{\eta_v(x_v+\delta_v x_u)} \\
\, \partial_{u,\sfp(u)}\Phi(x) = \, \partial_{\sfp(u),u}\Phi(x) &= -\frac{x_u}{x_{\sfp(u)}}\frac{w_u}{\eta_u(x_u+\delta_u x_{\sfp(u)})}.
\end{align*}

\subsection{Explicit dynamics}

We are now in a position to calculate the formal dynamics.
Let us define the control by $F(\cdot,t) \seteq -c(t)$.
We claim that for $u \neq \rt$,
\begin{equation}\label{eq:dyn2}
\partial_t \left(\frac{x_u(t)}{x_{\sfp(u)}(t)}\right) =  \frac{\eta_u}{w_u} \left(\frac{x_u(t)}{x_{\sfp(u)}(t)}+\delta_u\right)
\left(\beta_{\sfp(u)}(t)- \sum_{\ell \in \cL_u} \frac{x_{\ell}(t)}{x_{u}(t)} c_{\ell}\right),
\end{equation}
where $\beta_u(t)\ge 0$ denotes the Lagrange multiplier corresponding to the constraint $x_u = \sum_{v\in \chi(u)} x_v$.

To verify \eqref{eq:dyn2}, let us define, for $u \neq \rt$,
\[
\cE(u) \seteq
\frac{w_u}{\eta_u} \frac{x_{\sfp(u)}(t)}{x_u(t)+\delta_u x_{\sfp(u)}(t)} \partial_t \left(\frac{x_u(t)}{x_{\sfp(u)}(t)}\right).
\]
Then \eqref{eq:dyn2} is equivalent to the assertion that
\begin{equation}\label{eq:ass1}
   \cE(u) = \beta_{\sfp(u)}(t) - \sum_{\ell \in \cL_u} \frac{x_{\ell}(t)}{x_u(t)} c_{\ell}(t).
\end{equation}
Recalling \eqref{eq:mirror},
the equality $\left(\nabla^2 \Phi(x(t)) x'(t)\right)_u = \left(F(t,x(t))-\gamma(t)\right)_u$ is equivalent to
\begin{alignat}{3}
   \cE(\ell) &= \beta_{\sfp(\ell)}(t) - c_{\ell}(t), \qquad && \ell \in \cL, \label{eq:eleaves} \\
   \cE(u) - \sum_{v \in \chi(u)} \frac{x_v(t)}{x_u(t)} \cE(v) &= \beta_{\sfp(u)}(t) - \beta_u(t)\,, && u \in V \setminus (\cL \cup \{\rt\}). \label{eq:egen}
\end{alignat}
Clearly \eqref{eq:eleaves} already confirms \eqref{eq:ass1} for $\ell \in \cL$.

Let us conclude by verifying \eqref{eq:ass1} for all $u \notin \rt$  by (reverse) induction on the depth.
Employing \eqref{eq:egen} along with the validity of \eqref{eq:ass1} for $\{\cE(v) : v \in \chi(u)\}$ yields
\begin{align*}
   \cE(u) &= \beta_{\sfp(u)}(t) - \beta_u(t) + \sum_{v \in \chi(u)} \frac{x_v(t)}{x_u(t)} \left(\beta_u(t) - \sum_{\ell \in \cL_u} \frac{x_{\ell}(t)}{x_u(t)} c_{\ell}(t)\right) \\
          &= \beta_{\sfp(u)}(t) - \sum_{\ell \in \cL_u} \frac{x_{\ell}(t)}{x_{u}(t)} c_{\ell}(t),
\end{align*}
where we used the fact that $x_u = \sum_{v \in \chi(u)} x_v$ for $x \in \K_{T}$.

\subsection{Relationship between discrete and continuous dynamics}
\label{sec:continuous}

Recall the setup from \pref{sec:dyn}.
We consider a system of variables $\{q_u(t) : u \in V \setminus \{\rt\}\}$
satisfying the differential equations
\begin{equation}\label{eq:star-dynamics2}
  \partial_t q_u(t) = \frac{\eta_u}{w_u} \left(q_u(t)+\delta_u\right) \left(\beta_{\sfp(u)}(t) - \hat{c}_u(t) + \alpha_u(t)\right),
\end{equation}
where $\alpha_u(t)$ is a Lagrangian multiplier for the constraint $q_u(t)\ge 0$, and $\hat{c}_u(t)$ is the ``derived'' cost
in the subtree rooted at $u$:
\begin{align*}
\hat{c}_u(t) &\seteq \sum_{\ell \in \cL_u} q_{\ell\mid u}(t) c_{\ell}(t)\\
q_{\ell \mid u}(t) &\seteq \prod_{v \in \gamma_{u,\ell} \setminus \{u\}} q_v(t)\,,\nonumber
\end{align*}
where $\gamma_{u,\ell}$ is the unique simple $u$-$\ell$ path in $T$.
Now the values $q_{\ell \mid \vvr}$ give a probability distribution on the leaves.

Let us argue that when 
the discretization parameter of the
algorithm presented in \pref{sec:main} goes to zero,
one arrives at a solution to \eqref{eq:star-dynamics2}.
Recall that in \pref{sec:alg}, we split each cost function $c\in\R_+^{\cL}$
into $M$ pieces $M^{-1} c$ and computed a sequence of
configurations $q_0,\dots,q_M \in Q_T$.
Define the piecewise-linear function $q_{(M)} : [0,1] \to Q_T$ by
\[
   q_{(M)}\left(\frac{j+\delta}{M}\right) \seteq (1-\delta) q_{j} + \delta q_{j+1}, \qquad \delta \in [0,1], j \in \{0,\ldots,M-1\}.
\]

Recalling \pref{sec:iter}, we have
\begin{align}\label{eq:qeps}
   q_j^{(u)} &\seteq \argmin\left\{ \sfD^{(u)}\!\left(p \dmid q_{j-1}^{(u)}\right) +  \llangle p, M^{-1} \hat{c}_j^{(u)} \rrangle \Bigmid p \in
Q^{(u)}_T\right\},
\end{align}
where
\[
   \hat{c}_j^{(u)} = \sum_{\ell \in \cL_v} (q_j)_{\ell \mid u} c_{\ell}.
\]
Thus for $v \in \chi(u)$ and $j \geq 1$,
\[
   \left(q_j^{(u)}\right)_v = \left[\left(q_{j-1}^{(u)}\right)_v + \delta_v\right] \exp\left(\frac{\eta_v}{w_v} \left(\beta_u - (M^{-1}(\hat{c}^{(u)}_j)_v - \alpha_v)\right)\right) - \delta_v.
\]
One can now verify that there is a constant $L=L(c,T)$ such that
\[
   \left|\left(q_j\right)_v - \left(q_{j-1}\right)_v\right| \leq \frac{L}{M}, \quad j \in \{1,\ldots,M\}, v \in V \setminus \{\vvr\}.
\]
In particular, we see that $q'_{(M)} \in L^{\infty}([0,1],\R^{V \setminus \{\vvr\}})$ for every $M \geq 1$ and, moreover,
\begin{equation}\label{eq:derivbnd}
   \sup_{M \geq 1} \left\|q'_{(M)}\right\|_{L^{\infty}} < \infty.
\end{equation}
Therefore by Arzel\`a-Ascoli, there is a subsequence $\{M_k\}$ such that $q_{(M_k)}$ converges uniformly
to a function $q : [0,1] \to Q_T$.

Since the unit ball of $L^{\infty}([0,1], \R^{V \setminus \{\vvr\}})$ is weakly compact (by
the sequential Banach-Alaoglu Theorem), we can pass to a further subsequence $\{M'_{k}\}$ along which
$q'_{(M'_k)}$ converges weakly to some $h \in L^{\infty}([0,1], \R^{V \setminus \{\vvr\}})$.
Moreover,
since $q_{(M)}(b)-q_{(M)}(a) = \int_a^b q'_{(M)}(t)\,dt$ for all $0 \leq a < b \leq 1$, it follows that
$q(b)-q(a) = \int_a^b h(t)\,dt$ as well, and therefore
for almost all $t \in [0,1]$, we have $q'(t) = h(t)$.

If we similarly linearly interpolate the cost function to $\hat{c}_{(M)} : [0,1] \to \R_+^{V \setminus \{\vvr\}}$,
then $\hat{c}_{(M_k)} \to \hat{c}$ along this sequence as well,
and
\[
   \hat{c}^{(u)}(t) = \sum_{\ell \in \cL_v} q_{\ell \mid u}(t) c_{\ell}.
\]

Now the KKT conditions for optimality in \eqref{eq:qeps} give
\[
   \nabla \Phi^{(u)}\left(q_j^{(u)}\right) - \nabla \Phi^{(u)}\left(q_{j-1}^{(u)}\right) + M^{-1} \hat{c}^{(u)}_j \in - \sfN_{Q_T^{(u)}}\left(q_j^{(u)}\right),
\]
or equivalently,
\[
   \frac{\nabla \Phi^{(u)}\left(q_j^{(u)}\right) - \nabla \Phi^{(u)}\left(q_{j-1}^{(u)}\right)}{M^{-1}} \in -\hat{c}^{(u)}_j - \sfN_{Q_T^{(u)}}\left(q_j^{(u)}\right).
\]
By standard results in differential inclusion theory (e.g., the Convergence Theorem \cite[Thm. 1.4.1]{AC84}),
we conclude that $q : [0,1] \to Q_T$ solves the differential inclusion
\[
   \nabla^2 \Phi^{(u)}\!\left(q^{(u)}(t)\right) \partial_t q^{(u)}(t) \in -\hat{c}^{(u)}(t) - \sfN_{Q_T^{(u)}}(q^{(u)}(t)).
\]
Calculating the Hessian $\nabla^2 \Phi^{(u)}$ reveals that $q(t)$ is a solution to \eqref{eq:star-dynamics2}.

\section*{Acknowledgments}
Part of this work was carried out while C. Coester was visiting University of Washington, hosted by J. R. Lee. C. Coester was partially supported by EPSRC Award 1652110.  J. R. Lee was partially supported by NSF grants CCF-1616297 and CCF-1407779 and a Simons Investigator Award.

\bibliographystyle{alpha}
\bibliography{MTS}

\end{document}